\documentclass[12pt]{article}

\usepackage{a4wide, amsmath,amsthm,amsfonts,amscd,amssymb,eucal,bbm,mathrsfs, enumerate}

\usepackage{graphicx,tikz}
\usepackage[pdfborder={0 0 0}]{hyperref}
\usetikzlibrary{arrows}

\def\RR{{\mathbb R}}
\def\CC{{\mathbb C}}
\def\NN{{\mathbb N}}
\def\ZZ{{\mathbb Z}}

\def\A{{\mathcal A}}
\def\B{{\mathcal B}}
\def\C{{\mathcal C}}
\def\D{{\mathcal D}}

\def\F{{\mathcal F}}
\def\H{{\mathcal H}}

\def\K{{\mathcal K}}
\def\M{{\mathcal M}}
\def\N{{\mathcal N}}
\def\O{{\mathcal O}}

\def\R{{\mathcal R}}

\def\f{\varphi}

\def\L{\Lambda}

\def\o{\omega}

\def\r{\rho}

\def\t{\tau}

\def\x{\xi}

\def\Ad{{\hbox{\rm Ad\,}}}

\def\1{{\mathbbm 1}}

\def\ran{{\rm Ran}\,}

\def\uone{{\rm U(1)}}
\def\u1net{{\A^{(0)}}}

\def\conf{{\rm Conf(\E)}}
\def\diff{{\rm Diff}}

\def\diffs1{\diff(S^1)}
\def\mob{{\rm M\ddot{o}b}}
\def\mob2{{\rm M\ddot{o}b}^{(2)}}

\def\supp{{\rm supp\,}}

\def\slim{{{\mathrm{s}\textrm{-}\lim}\,}}

\def\timesi{{\overset{\tin}\times}}
\def\timeso{{\overset{\tout}\times}}
\def\psl2r{{\rm PSL}(2,\RR)}
\def\sl2r{{\rm SL}(2,\RR)}
\def\su11{{\rm SU}(1,1)}
\def\2dmob{{\overline{\psl2r}\times\overline{\psl2r}}}
\def\<{\langle}
\def\>{\rangle}
\def\Re{\mathrm{Re}\,}
\def\Im{\mathrm{Im}\,}

\def\hout{\H^\tout}
\def\houtprod{\H^\tout_{\mathrm{prod}}}
\def\timeso{{\overset{\mathrm{out}}{\times}}}

\def\conf{{\mathscr C}}
\def\gconf{{\widetilde{\mathscr C}}}
\def\cmink{{\bar M}}
\def\cyl{{\widetilde{M}}}

\def\cc{{\mathrm c}}

\newcommand{\Om}{\Omega} 
\newcommand{\phiout}{\Phi^{\mathrm{out}}}
\newcommand{\tout}{\mathrm{out}}
\newcommand{\tin}{\mathrm{in}}
\newcommand{\tdir}{\mathrm{dir}}

\newtheorem{theorem}{Theorem}[section]

\newtheorem{corollary}[theorem]{Corollary}
\newtheorem{proposition}[theorem]{Proposition}
\newtheorem{lemma}[theorem]{Lemma}
\theoremstyle{remark}

\title{Massless Wigner particles in conformal field theory are free}
\date{}
\author{
{\bf Yoh Tanimoto} \footnote{Supported by
Alexander von Humboldt Stiftung until March 2013.} \\
e-mail: {\tt hoyt@ms.u-tokyo.ac.jp}\\
Graduate School of Mathematical Sciences, The University of Tokyo\\
and Institut f\"ur Theoretische Physik, G\"ottingen University\\
3-8-1 Komaba Meguro-ku Tokyo 153-8914, Japan.\\
JSPS SPD postdoctoral fellow\\
}
\begin{document}
\maketitle
\begin{abstract}
We show that in a four dimensional conformal Haag-Kastler net, its massless particle
spectrum is generated by a free field subnet. If the massless particle spectrum
is scalar, then the free field subnet decouples as a tensor product component.
\end{abstract}

\section{Introduction}\label{introduction}
Conformal field theories have been extensively studied in two-dimensional spacetime.
There are many examples, certain exact computations are available and they provide
also interesting mathematical structures.
On the other hand, from a mathematical point of view, no nonperturbative construction
of a single interacting quantum field theory in four dimensional spacetime
is available today. In this paper, instead
of constructing models, we try to understand general restrictions on models with a large spacetime
symmetry.
We prove that if a conformal field theory in four spacetime dimensions
in the operator-algebraic approach (Haag-Kastler net) contains massless particles, then
there is a free subnet generating the massless particles. Furthermore, if the massless
particles are scalar, then they decouple as a tensor product component.
Therefore, massless particles in conformal field theory cannot interact.

Actually Buchholz and Fredenhagen have already proved more than 30 years ago
that the S-matrix of a dilation-invariant theory is trivial \cite{BF77}.
Based on this result, Baumann \cite{Baumann82} has shown that any dilation-invariant
scalar field (in the sense of Wightman) where a complete particle interpretation is available
(asymptotic completeness with respect to massless particles) is the Wick product of the free field.
Compared to these, our results are not necessarily
stronger because we assume conformal invariance. On the other hand, there are more general aspects:
our framework is Haag-Kastler nets and we do not assume neither the existence of Wightman fields,
nor asymptotic completeness.
In two-dimensional spacetime, triviality of S-matrix does not
necessarily imply that the net is free (second quantized). Indeed, in our previous work
\cite{Tanimoto12-1}, we have seen that a two-dimensional conformal net is asymptotically
complete with respect to massless waves if and only if it is the tensor product of its chiral components.
Hence one may consider the tensor product subnet as the ``particle-like'' (or ``wave-like'') part.
However, chiral components can be highly nontrivial (different from the second quantized net,
the $\uone$-current net). In comparison, in four dimensions,
we prove that the particle spectrum is generated by the free, second quantized net.
In particular, if the particles are scalar, the free field subnet which we construct
cannot have any nontrivial extension,
hence it must decouple in the full net. This is the operator-algebraic version
of the argument given in \cite[Section 1]{BNRT07}.
Relaxing the assumption of asymptotic completeness (with respect to massless particles) is important, because
while there are many physical arguments that dilation-invariance should imply
conformal invariance \cite{Nakayama13, DKSS13}, conformal field theory
may contain massive spectrum (the meaning of ``massive'' will be clarified in Section \ref{representation}),
as one would expect from the maximally supersymmetric Yang-Mills theory, which should be
conformal \cite{Mandelstam83}.

We stress that our approach is nonperturbative. We make an assumption that there is
a nonperturbatively given model as a conformal Haag-Kastler net.
The existence of massless particles \`a la Wigner is defined in the sense that
the representation of the spacetime translations has nontrivial spectral projection
on the surface of the positive lightcone.
In this case, Buchholz has established the existence of asymptotic fields \cite{Buchholz77}.
Besides, operator-algebraic scattering theory has been successfully
applied to many massive models in low dimensions.
The theory was able to reconstruct the factorizing S-matrix as an invariant of the net
\cite{Lechner08, Tanimoto13-1}.

There are more claims that conformal fields with massless particles
are free with different assumptions \cite{Todorov06, Weinberg12}.
An advantage of our approach is to avoid any field-theoretic calculation.
One of the main tools is the Tomita-Takesaki modular theory applied to
conformal nets \cite{BGL93}: Brunetti, Guido and Longo have shown that the modular group
of a double cone is certain conformal transformations which preserve the double cone.
This renders the central idea of our arguments geometric, combined with
the construction of asymptotic fields by Buchholz \cite{Buchholz77}.

Let us recall a technical conjecture in \cite{Buchholz77}.
In order to obtain asymptotic fields, one had to choose local operators
with a certain regularity condition in the momentum space,
although Buchholz conjectured that this construction should extend to any local operator.
In our application, this restriction is a problem because the regularity condition is not stable
under conformal transformations. We remove this restriction and show that
the asymptotic fields are covariant under the conformal transformation of
the given net.

This paper is organized as follows. In Section \ref{preliminaries} we summarize
the foundations of conformal nets and the massless scattering theory.
The technical conjecture above is proved there.
We first state and prove our results on the existence of free subnet
for globally conformal nets in Section \ref{gci}.
This additional assumption greatly reduces the problem and emphasizes
the geometric nature of our proof.
Section \ref{general} treats the general case, not necessarily globally conformal
but conformal. We also prove the decoupling of the free scalar subnet. 
Finally we discuss open problems and future directions in Section \ref{open}.

\section{Preliminaries}\label{preliminaries}
\subsection{Conformal field theory}
A model of quantum field theory is realized
as a net of von Neumann algebras. A conformal field theory is a net with
the conformal symmetry. We collect here the definitions and results
necessary for our analysis.

\subsubsection{The conformal group and the extended Minkowski space}\label{conformalgroup}
We consider $\RR^4$, the Minkowski space. A conformal symmetry
is a transformation of $\RR^4$ which preserves the Lorentz metric
$a\cdot b = a_0b_0 - \sum a_kb_k$
up to a function. Actually we allow a symmetry to take a meager set
out of $\RR^4$. Hence we need to consider local actions, following
the work by Brunetti-Guido-Longo \cite{BGL93}.

Let $G$ be a Lie group and $M$ be a manifold.
We say that $G$ {\bf acts locally on} $M$ if there is an open nonempty set
$B \subset G \times M$ and a smooth map $T: B\to M$ such that
\begin{enumerate}[{(} 1 {)}]
 \item For any $a \in M$ , $V_a := \{g \in G: (g,a) \in B\}$ is an open connected neighborhood of the
 unit element $e$ of $G$.
 \item $T_ea = a$ for any $a \in M$.
 \item For $(g,a) \in B$, it holds that $V_{T_ga} = V_ag^{-1}$ and for $h \in G$ such that
 $hg \in V_a$, one has $T_hT_ga = T_{hg}a$.
\end{enumerate}

In the following, we only consider $M = \RR^4$.
The {\bf conformal group} $\conf$ is generated by the Poincar\'e group,
dilations and the special conformal transformations: a special conformal
transformation is of the form $\r \t(a) \r$, where $\t(a)$ is a translation
by $a \in \RR^4$ and $\r$ is the relativistic ray inversion
\[
 \r a = -\frac{a}{a\cdot a}.
\]
This action is quasi global in the sense that
for any $g \in \conf$ the open set $\{a\in M: (g,a)\in B\}$ is the complement of a meager set $S_g$
and it holds for $a_0 \in S_g$ that $\lim_{a\to a_0} T_ga = \infty$. In other words,
the set of points in $M$ which are taken out of $M$ by $g$ is meager.
This action $T$ is transitive. It has been shown \cite[Propositions 1.1, 1.2]{BGL93}
that there is a manifold $\cmink$ such that $M$ is a dense open subset of $\cmink$
and the action $T$ extends to a transitive global action on $\cmink$.
Furthermore, the action of $T$ lifts to a transitive global action $\widetilde T$
of the universal covering group $\widetilde G$ of $G$ on the universal covering
$\cyl$ of $\cmink$.

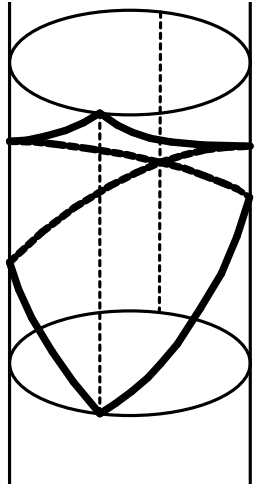
\begin{figure}[ht]
    \centering
\begin{tikzpicture}[scale=0.8, line cap=round,line join=round,>=triangle 45,x=1.0cm,y=1.0cm]
\clip(-4,-2) rectangle (2,6);
\draw [line width=1.2pt] (-3,-2) -- (-3,6);
\draw [line width=1.2pt] (1,-2) -- (1,6);
\draw [rotate around={0:(-1,5)},line width=1.2pt] (-1,5) ellipse (2cm and 0.87cm);
\draw [rotate around={0:(-1,0)},line width=1.2pt] (-1,0) ellipse (2cm and 0.87cm);
\draw [line width=1.2pt,dash pattern=on 2pt off 2pt] (-0.49,5.84)-- (-0.51,0.85);
\draw [line width=1.2pt,dash pattern=on 2pt off 2pt] (-1.5,-0.84)-- (-1.5,4.16);
\draw [line width=2.8pt] (-1.5,4.16)-- (-1.77,4);
\draw [line width=2.8pt] (-1.77,4)-- (-2.07,3.87);
\draw [line width=2.8pt] (-2.07,3.87)-- (-2.39,3.78);
\draw [line width=2.8pt] (-2.39,3.78)-- (-2.71,3.71);
\draw [line width=2.8pt] (-2.71,3.71)-- (-3,3.69);
\draw [line width=2.8pt,dotted] (-3,3.69)-- (-2.58,3.68);
\draw [line width=2.8pt,dotted] (-2.58,3.68)-- (-2.22,3.64);
\draw [line width=2.8pt,dotted] (-2.22,3.64)-- (-1.69,3.57);
\draw [line width=2.8pt,dotted] (-1.69,3.57)-- (-1.06,3.47);
\draw [line width=2.8pt,dotted] (-1.06,3.47)-- (-0.5,3.34);
\draw [line width=2.8pt,dotted] (-0.5,3.34)-- (-0.09,3.24);
\draw [line width=2.8pt,dotted] (-0.09,3.24)-- (0.23,3.12);
\draw [line width=2.8pt,dotted] (0.23,3.12)-- (0.54,3);
\draw [line width=2.8pt,dotted] (0.54,3)-- (0.78,2.9);
\draw [line width=2.4pt,dotted] (0.78,2.9)-- (0.99,2.76);
\draw [line width=2.8pt] (-1.5,-0.84)-- (-1.24,-0.67);
\draw [line width=2.8pt] (-1.24,-0.67)-- (-0.97,-0.47);
\draw [line width=2.8pt] (-0.97,-0.47)-- (-0.7,-0.24);
\draw [line width=2.8pt] (-0.7,-0.24)-- (-0.49,-0.01);
\draw [line width=2.8pt] (-0.49,-0.01)-- (-0.21,0.32);
\draw [line width=2.8pt] (-0.21,0.32)-- (0.17,0.91);
\draw [line width=2.8pt] (0.17,0.91)-- (0.52,1.49);
\draw [line width=2.8pt] (0.52,1.49)-- (0.78,2.11);
\draw [line width=2.8pt] (0.78,2.11)-- (0.99,2.76);
\draw [line width=2.8pt] (-1.5,-0.84)-- (-1.94,-0.31);
\draw [line width=2.8pt] (-1.94,-0.31)-- (-2.29,0.19);
\draw [line width=2.8pt] (-2.29,0.19)-- (-2.63,0.75);
\draw [line width=2.8pt] (-2.63,0.75)-- (-2.85,1.23);
\draw [line width=2.8pt] (-2.85,1.23)-- (-3,1.68);
\draw [line width=2.8pt,dotted] (-3,1.68)-- (-2.46,2.18);
\draw [line width=2.8pt,dotted] (-2.46,2.18)-- (-1.92,2.59);
\draw [line width=2.8pt,dotted] (-1.92,2.59)-- (-1.33,2.95);
\draw [line width=2.8pt,dotted] (-1.33,2.95)-- (-0.89,3.17);
\draw [line width=2.8pt,dotted] (-0.89,3.17)-- (-0.5,3.34);
\draw [line width=2.8pt,dotted] (-0.5,3.34)-- (-0.11,3.46);
\draw [line width=2.8pt,dotted] (-0.11,3.46)-- (0.26,3.55);
\draw [line width=2.8pt,dotted] (0.26,3.55)-- (0.59,3.59);
\draw [line width=2.8pt,dotted] (0.59,3.59)-- (1,3.61);
\draw [line width=2.8pt] (-1.5,4.16)-- (-1.21,3.96);
\draw [line width=2.8pt] (-1.21,3.96)-- (-0.96,3.84);
\draw [line width=2.8pt] (-0.96,3.84)-- (-0.66,3.74);
\draw [line width=2.8pt] (-0.66,3.74)-- (-0.32,3.68);
\draw [line width=2.8pt] (-0.32,3.68)-- (0.21,3.63);
\draw [line width=2.8pt] (0.21,3.63)-- (1,3.61);
\end{tikzpicture}
    \caption{The global space $\cyl$ projected on the two-dimensional cylinder. The region surrounded by thick lines is a copy of the Minkowski space.}
    \label{fig:minkowski}
\end{figure}

We can realize $\cmink$ concretely in $\RR^6$ as follows:
\[
 N := \{(\xi_0,\cdots, \xi_5) \in \RR^6 \setminus \{0\}: \xi_0^2-\xi_1^2-\cdots-\xi_4^2+\xi_5^2 = 0\}/\RR^*,
\]
where $\RR^* = \RR\setminus \{0\}$ acts on $\RR^6$ by multiplication.
For $a \in M = \RR^4$, we define the embedding by
 $\xi_k = a_k$ for $k = 0,1,2,3$ and $\xi_4 = \frac{1-a\cdot a}{2}, \xi_5 = \frac{1+a\cdot a}{2}$.
The group $\mathrm{PSO}(4,2)$ acts on $N$ and this corresponds to the action
of the conformal group $\conf$. Since the image of $M$ in $N$ is dense,
it follows that $N = \cmink$ \cite{BGL93}.
One observes that $N$ is diffeomorphic to $(S^3 \times S^1) / \ZZ_2$, hence its
universal covering is $S^3 \times \RR$.

\subsubsection{Conformal nets}\label{conformalnet}
An operator-algebraic conformal field theory, or a {\bf conformal net},
is a triple $(\A, U, \Om)$  of a map $\A$ from the family of open double cones in
$M$ into the family of von Neumann algebras on $\H$, a local unitary representation
(the group structure is respected only locally) $U$
of the conformal group $\conf$ and a unit vector $\Om \in \H$ such that
\begin{enumerate}[{(}1{)}]
\item {\bf Isotony.} If $O_1 \subset O_2$, then $\A(O_1) \subset \A(O_2)$.
\item {\bf Locality.} If $O_1$ and $O_2$ are spacelike separated, then $\A(O_1)$ and $\A(O_2)$ commute.
\item {\bf Local conformal covariance.} For each double cone $O \subset M$,
there is a neighborhood $V_O$ of the identity of $\conf$ such that
$V_O \times O \subset B$, where $B$ is the domain of the local action of $\conf$
on $M$, such that $\Ad U(g) (\A(O)) = \A(gO)$.
\item {\bf Positivity of energy.} The spectrum of the subgroup of
translations in $\conf$ in the representation $U$
(this is well-defined although the action $U$ is local, since
the group of translations is simply connected)
is included in the
closed positive lightcone $\overline V_+ := \{a \in \RR^4: a_0 \ge 0, \; a\cdot a \ge 0\}$.
\item {\bf Vacuum.} The vector $\Om$ is invariant under the action of $U$.
Such a vector is unique up to a scalar.
\item {\bf Reeh-Schlieder property.} The vector $\Om$ is cyclic and separating
 for each local algebra $\A(O)$.
\end{enumerate}
Note that Reeh-Schlieder property is usually proved under additivity.
We take it here as an assumption for simplicity
(see the discussion in \cite[Section 2]{Weiner11}).

A conformal net can be extended to $\cyl$ with the action of
$\gconf$ \cite[Proposition 1.9]{BGL93}. Indeed, the representation
$U$ lifts to $\gconf$ and the local algebra $\A(O)$ for
$O$ which is not included in $\cyl$ is defined by covariance.

A {\bf (conformal) subnet} $\A_0$ of a net $(\A,U,\Om)$ is
a family of von Neumann subalgebras $\A_0(O) \subset \A(O)$
such that isotony and covariance with respect to the same $U$ hold.
In this case, $\overline{\A_0(O)\Om}$ is a Hilbert subspace of $\H$ independent of $O$.

\subsubsection{Bisognano-Wichmann property}\label{bisognano-wichmann}
Certain regions play a special role in the study of conformal field theory.
Here we pick the standard wedge in the $a_1$-direction, the unit double cone
and the future lightcone:
\begin{itemize}
 \item $W_1 := \{a \in M: a_1 > |a_0|\}$,
 \item $O_1 := \{a \in M: |a_0|+\sqrt{a_1^2 + a_2^2 + a_3^2}  < 1\}$,
 \item $V_+ := \{a \in M: a_0 > 0, \; a\cdot a > 0\}$
\end{itemize}
To each of these regions $O$ in $\cyl$ we associate a one-parameter group
$\L^O_t$ in $\gconf$ which preserve $O$ and commute with all $O$-preserving
conformal transformations:
\begin{itemize}
 \item For the wedge $W_1$, we take the boosts in $a_1$-direction.
 They are linear transformations and their actions on $(a_0,a_1)$ components can be written, in a matrix form,
 as $\L^{W_1}_t = \left(\begin{array}{cc} \cosh 2\pi t & -\sinh 2\pi t \\ -\sinh 2\pi t & \cosh 2\pi t\end{array}\right)$.
 \item For the unit double cone, by rotation invariance the action is determined by the action on $(a_0,a_1)$-plane:
\[
 \L^{O_1}_ta_\pm = \frac{(1 + a_\pm) - e^{-2\pi t}(1 - a_\pm)}{(1 + a_\pm) - e^{-2\pi t}(1 + a_\pm)},
\]
 where $a_\pm = a_0 \pm a_1$.
 \item For the future lightcone $V_+$, we take the dilation:
 $\L^{V_+}_t a = e^{2\pi t} \cdot a$.
\end{itemize}
These regions are mapped to each other by conformal transformations (on $\cyl$) and
the associated transformations are coherent, in the sense that
$\L^{O}_t = g^{-1}\L^{O'}_tg$ where $O = gO'$, $g \in \gconf$ and
$O, O' = W_1, O_1, V_+$. One can define $\L^{O}_t$ for any other double cone, wedge or
lightcone by coherence.

For a conformal net, the modular group of a local algebra with respect to
the vacuum has been completely determined \cite{BGL93}.
\begin{theorem}[Bisognano-Wichmann property]
 Let $(\A, U, \Om)$ be a conformal net and consider its natural extension to $\cyl$.
 Then for any image $O$ of a double cone by a conformal transformation in $\gconf$,
 one has $\Delta_O^{it} = U(\L^O_t)$, where $\Delta_O$ is the modular operator of $\A(O)$ with
 respect to $\Om$.
\end{theorem}

The following duality has been also proved \cite{BGL93}.
\begin{theorem}[Haag duality on $\cyl$]
 Let $(\A, U, \Om)$ be a conformal net and consider its natural extension to $\cyl$.
 Then for a wedge $W$, it holds that $\A(W)' = \A(W')$.
\end{theorem}
Since a conformal transformation can bring a wedge to a double cone $O$, a similar duality
holds for double cones. In that case, we need the causal complement $O^\cc$ on $\cyl$ rather than
the usual spacelike complement $O'$.

\begin{figure}[ht]
    \centering
\begin{tikzpicture}[scale=0.6, line cap=round,line join=round,>=triangle 45,x=1.0cm,y=1.0cm]
\clip(-5,-5) rectangle (10,7);
\fill[fill=black,fill opacity=1.0] (-1,1) -- (1,-1) -- (3,1) -- (1,3) -- cycle;
\fill[line width=1.6pt,fill=black,fill opacity=0.5] (4,2) -- (3,1) -- (4,0) -- (5,1) -- cycle;
\fill[line width=1.6pt,fill=black,fill opacity=0.35] (4,0) -- (6,-2) -- (7,-1) -- (5,1) -- cycle;
\fill[fill=black,fill opacity=0.5] (-2,2) -- (-2,0) -- (-1,1) -- cycle;
\fill[fill=black,fill opacity=0.5] (7,3) -- (5,1) -- (7,-1) -- (8,0) -- (8,2) -- cycle;
\fill[fill=black,fill opacity=0.25] (4,2) -- (5,1) -- (7,3) -- (6,4) -- cycle;
\draw [line width=1.6pt] (8,-5) -- (8,7);
\draw [line width=2.4pt] (0,6)-- (8,-2);
\draw (-1,1)-- (1,3);
\draw (1,3)-- (3,1);
\draw (3,1)-- (1,-1);
\draw (1,-1)-- (-1,1);
\draw (-1,1)-- (1,-1);
\draw (1,-1)-- (3,1);
\draw (3,1)-- (1,3);
\draw (1,3)-- (-1,1);
\draw [line width=1.6pt] (4,2)-- (3,1);
\draw [line width=1.6pt] (3,1)-- (4,0);
\draw [line width=1.6pt] (4,0)-- (5,1);
\draw [line width=1.6pt] (5,1)-- (4,2);
\draw [line width=1.6pt] (4,0)-- (6,-2);
\draw [line width=1.6pt] (6,-2)-- (7,-1);
\draw [line width=1.6pt] (7,-1)-- (5,1);
\draw [line width=1.6pt] (5,1)-- (4,0);
\draw [line width=2.4pt] (0,-4)-- (8,4);
\draw [line width=2.4pt] (0,-4)-- (-2,-2);
\draw [line width=2.4pt] (-2,4)-- (0,6);
\draw [line width=1.6pt] (-2,-5) -- (-2,7);
\draw [line width=1.6pt] (-1,1)-- (-2,2);
\draw [line width=1.6pt] (-1,1)-- (-2,0);
\draw [line width=1.6pt] (7,-1)-- (8,0);
\draw [line width=1.6pt] (8,2)-- (7,3);
\draw [line width=1.6pt] (7,3)-- (6,4);
\draw [line width=1.6pt] (6,4)-- (4,2);
\draw (-2,2)-- (-2,0);
\draw (-2,0)-- (-1,1);
\draw (-1,1)-- (-2,2);
\draw (7,3)-- (5,1);
\draw (5,1)-- (7,-1);
\draw (7,-1)-- (8,0);
\draw (8,0)-- (8,2);
\draw (8,2)-- (7,3);
\draw (4,2)-- (5,1);
\draw (5,1)-- (7,3);
\draw (7,3)-- (6,4);
\draw (6,4)-- (4,2);
\end{tikzpicture}
    \caption{Regions in the global space $\cyl$. The left and right sides are identified.
    The white square: a copy of the Minkowski space. Black: a double cone $O$.
    Dark gray: the spacelike complement $O'$ of the double cone in the Minkowski space.
    Light gray + dark gray: the causal complement $O^\cc$ in $\cyl$.}
    \label{fig:duality}
\end{figure}
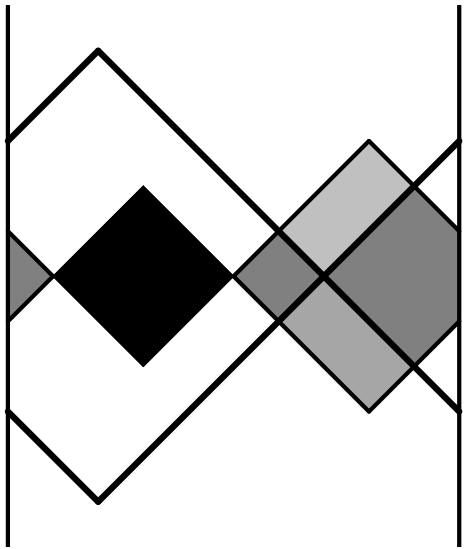

\subsubsection{Representation theory of the conformal group}\label{representation}
The conformal group is locally isomorphic to $\mathrm{SU}(2,2)$
and its unitary positive-energy irreducible representations have been classified \cite{Mack77}.
Using the dimension $d \ge 0$ and half-integers $j_1,j_2 \ge 0$, they are parametrized as follows.
When restricted to the Poincar\'e group, one can consider the mass parameter $m$
and spin $s$ or helicity.
\begin{itemize}	
 \item trivial representation. $d=j_1=j_2=0$.
 \item $j_1\neq 0 \neq j_2$, $d > j_1+j_2+2$. In this case, $m > 0$ and $s = |j_1-j_2|, \cdots j_1+j_2$ (integer steps).
 \item $j_1j_2=0$, $d > j_1+j_2+1$. $m > 0$ and $s = j_1+j_2$.
 \item $j_1 \neq 0 \neq j_2$, $d = j_1+j_2+2$. $m > 0$ and $s = j_1+j_2$.
 \item $j_1j_2 = 0$, $d=j_1+j_2+1$. $m = 0$ and helicity $s = j_1-j_2$.
\end{itemize}
Hence, the only massless representations are the last family.
In this paper, when we say that a conformal net contains massless particles,
it means that the representation $U$ has a subrepresentation in this family.

In \cite{Weinberg12} the following has been proved: if there is a quantum field
(an operator-valued distribution) which transforms as a vector in one of the
above massless representations,
then it is free. It implicitly assumes that the massless particles are
generated by such a field. This is apparently a stronger assumption
than the one in the operator-algebraic approach (see Section \ref{scattering})
that local observables generate states which contain massless particles.

The other nontrivial representations have mass $m>0$. One can call them massive,
although there is no mass gap because of the action of dilations.

\subsection{Massless scattering theory}\label{scattering}
In the operator-algebraic approach, the concept of particle is not given a priori, but
to be defined through operational process.
Such a theory for massless particles has been established in \cite{Buchholz77}
for a Poincar\'e covariant net
under the assumption that the representation of the translation has nontrivial
spectral projection corresponding to the cone $m = 0$.
In such a case, we say that the net contains massless particles (following Wigner).

\subsubsection{Convergence of asymptotic fields for regular operators}\label{convergence}
Let $(\A, U, \Om)$ be a Poincar\'e covariant net (a net for which the covariance
is only assumed for the Poincar\'e group).
Let $x$ be an operator in $\A(O)$ which is smooth in norm under the group action
$g \mapsto \Ad U(g)(x)$.
There are sufficiently many such operators. Indeed, if $x$ is localized in
a slightly smaller region than $O$, then
one can smear $x$ with a smooth function with compact support in the group
(note that the conformal group $\conf$ is finite-dimensional).
For a vector $a \in M$, we denote $x(a) = \Ad U(\t(a))(x)$.
For $t \in \RR$, we define
\[
 \Phi^t(x) := -2t\int_{S^2} d\o(\mathbf n)\; \partial_0 x(t,t \mathbf n),
\]
where $d\o$ is the normalized rotation-invariant measure on $S^2$ and
$\partial_0$ is the derivative with respect to the time translation (which is independent from $t$).
By a straightforward calculation, one finds that
\[
  \Phi^t(x)\Om =
  \frac{1}{|\mathbf{P}|}(e^{it(H-|\mathbf{P}|)} - e^{it(H+|\mathbf{P}|)})Hx\Om,
\]
where $P = (H,\mathbf{P})$ is the generator of translation: $U(\t(a)) = e^{itP\cdot a}$.
Furthermore, we need to take suitable time-averages. 
We fix a positive, smooth and compactly supported function $h$ with $\int_{\RR} h(t)dt = 1$ and
$h_T(t) = \frac1{\log |T|}\, h\left(\frac{t-T}{\log |T|}\right)$.
We set
\[
 \Phi^{h_T}(x) = \int_{\RR} dt\; h_T(t)\Phi^t(x).
\]
Then by the mean ergodic theorem one obtains \cite{Buchholz82}
\[
 \underset{T\to \infty}{\slim} \Phi^{h_T}(x)\Om = P_1
 x\Om,
\]
where $P_1$ is the projection onto the massless one-particle space, where
$H = |\mathbf P|$ holds.

For any double cone $O$, we denote by $V_{O,+}$ the future tangent of $O$,
the set of all points separated by a future-timelike vector from any point of $O$.
For a fixed double cone $O_+$ in $V_{O,+}$, there is a sufficiently large $T$ such that
$\Phi^{h_T}(x)$ is contained in the causal complement of $O_+$.
In particular, for sufficiently large $T$, there is a large commutant
for $\Phi^{h_T}(x)$ and one can define the operator $\phiout(x)$ by
$\phiout(x)y\Om = \underset{T\to\infty}\slim y \Phi^{h_T}(x)\Om = yP_1x\Om$,
where $y \in \A(O_+)$.
Let us denote $\F(V_{O,+}) = \bigcup_{O_+\subset V_{O,+}} \A(O_+)$
(the union, not the weak closure and $O_+$ are bounded).
The choice of $O_+$ was arbitrary in $V_{O,+}$, hence
$\phiout(x)$ can be defined on $\F(V_{O,+})\Om$. It is easy to see that $\phiout(x)$ is closable.
We denote the closure by the same symbol and its domain by $\D(\phiout(x))$.

For $N \in \NN$, let $\A_N(O)$ be the linear span of the operators
\[
 \int_{\RR} dt\; \f(t)\Ad U(\t(ta))(x),
\]
where $x \in \A(\check O)$, $a$ is a timelike vector and $\f$ is a test function with compact support
which has a Fourier transform $\tilde \f(p)$ with an $N$-fold zero at $p=0$,
and $\check O + (\supp \f )a\subset O$.

\begin{figure}[ht]
    \centering
\begin{tikzpicture}[scale=0.5, line cap=round,line join=round,>=triangle 45,x=1.0cm,y=1.0cm]
\clip(-15,-9) rectangle (15,9);
\draw [fill=black,fill opacity=0.75] (1,-6) circle (1.41cm);
\fill[fill=black,fill opacity=0.1] (1,2) -- (13,-10) -- (-11,-10) -- cycle;
\fill[fill=black,fill opacity=0.1] (1,2) -- (18,19) -- (-16,19) -- cycle;
\draw [rotate around={0:(1,6.5)},line width=1.6pt] (1,6.5) ellipse (4.26cm and 1.46cm);
\draw [rotate around={0:(1,-2.5)},line width=1.6pt] (1,-2.5) ellipse (4.26cm and 1.46cm);
\draw [line width=1.6pt] (-3,-2)-- (5,6);
\draw [line width=1.6pt] (-3,6)-- (5,-2);
\draw [line width=1.6pt,dash pattern=on 6pt off 6pt,domain=-15.0:-2.9999999999999947] plot(\x,{(-3--1*\x)/-1});
\draw [line width=1.6pt,dash pattern=on 6pt off 6pt,domain=5.0:15.0] plot(\x,{(--1--1*\x)/1});
\draw [line width=1.6pt,dash pattern=on 6pt off 6pt,domain=-15.0:-3.0] plot(\x,{(-1-1*\x)/-1});
\draw [line width=1.6pt,dash pattern=on 6pt off 6pt,domain=5.0:15.0] plot(\x,{(--3-1*\x)/1});
\draw [line width=1.2pt,dotted,domain=1.0:15.0] plot(\x,{(-9--1*\x)/1});
\draw [line width=1.2pt,dotted,domain=-15.0:1.0] plot(\x,{(--7--1*\x)/-1});
\draw [line width=1.2pt,dotted,domain=-15.0:2.0] plot(\x,{(--1.5--0.5*\x)/-0.5});
\draw [line width=1.2pt,dotted,domain=0.0:15.0] plot(\x,{(-2.5--0.5*\x)/0.5});
\draw [rotate around={0:(1,2)},line width=1.2pt,dotted] (1,2) ellipse (5.52cm and 2.35cm);
\draw (1,2)-- (13,-10);
\draw (13,-10)-- (-11,-10);
\draw (-11,-10)-- (1,2);
\draw (1,2)-- (18,19);
\draw (18,19)-- (-16,19);
\draw (-16,19)-- (1,2);
\draw [shift={(1,-3)},line width=0.4pt]  plot[domain=4.28:5.15,variable=\t]({1*3.36*cos(\t r)+0*3.36*sin(\t r)},{0*3.36*cos(\t r)+1*3.36*sin(\t r)});
\end{tikzpicture}
    \caption{How asymptotic fields are constructed. A local observable in a dark gray region is taken in the region between the cones indicated by dotted lines.}
    \label{fig:asymtotic}
\end{figure}
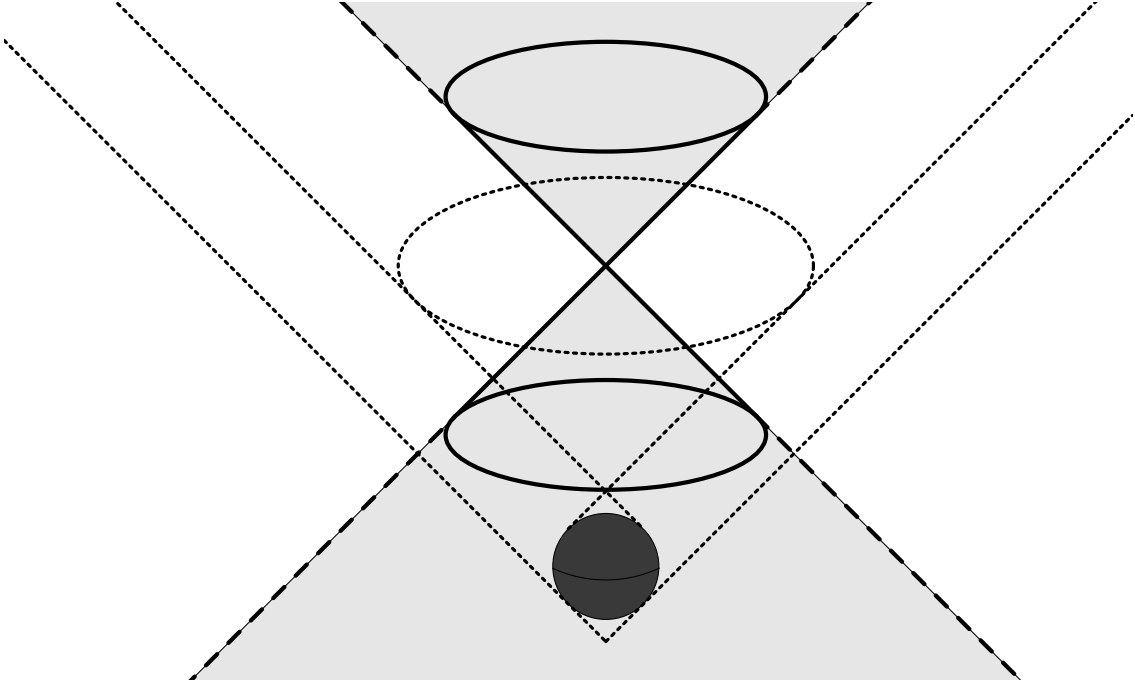

The following has been proved \cite[Lemma 1, Lemma 6, Theorems 7, 8, 9]{Buchholz77}.
\begin{theorem}[Buchholz]\label{th:bu}
Let $x = x^*$ be an element of $\A_{N_0}(O)$, where $N_0 \ge 15$,
$O$ is a double cone and $V_{O,+}$ be the future tangent of $O$.
Then the following hold.
\begin{enumerate}[{(}1{)}]
 \item\label{lee:commutation} For an arbitrary $y \in \A(O_+)$, where $O_+ \subset V_{O,+}$ is bounded,
 $y\cdot \D(\phiout(x))\subset \D(\phiout(x))$ and
 one has $[\phiout(x), y] = 0$ on $\D(\phiout(x))$.
 \item\label{lee:core} The operator $\phiout(x)$ is self-adjoint and depends only on
 $P_1 x\Om$. The subspace $\F(V_{O,+})\Om$ is a core of $\phiout(x)$.
 \item\label{lee:converngece} The sequence $\Phi^{h_T}(x)$ is convergent to $\phiout(x)$ in the
 strong resolvent sense.
 \item\label{lee:houtprod} The operator $\phiout(x)$ can be applied to the vacuum $\Omega$
 arbitrarily many times. We denote the vectors generated in this way recursively (the first
 term in the right-hand side which contains $n+1$ product is defined in this way):
 \[
  \phiout(x)\cdot \xi_1\timeso\xi_2\timeso \cdots \timeso \xi_n =
  \xi\timeso\xi_1\timeso\xi_2\timeso \cdots \timeso \xi_n +
  \sum_{k=1}^n \<\xi,\xi_k\> \xi_1\timeso\cdots \check{\xi}_k\cdots \timeso\xi_n,
 \]
 where $\xi = P_1x\Om = P_1x^*\Om$ and $\check{\xi}_k$ means the omission of the $k$-th element.
 Then the symbol $\timeso$ is compatible (unitarily equivalent) with
 the normalized symmetric tensor product on the Fock space with the one particle space $P_1\H$.
 The domain of $\phiout(x)$ includes the set $\houtprod$ of all linear combinations (without closure)
 of product states
 $\xi_1\timeso\xi_2\timeso \cdots \timeso \xi_n$, where $\xi_k$ is an arbitrary vector in $P_1\H$.
 \item\label{lee:covariance} It holds that $\Ad U(g)(\phiout(x)) = \phiout(\Ad U(g)(x))$
 if $g$ is a Poincar\'e transformation.
 \item For the resolvent $R_{\pm i}(y)= (y\pm i)^{-1}$ of $y$, it holds that
\begin{align*}
&[R_{\pm i}(\phiout(x_1)), R_{\pm i}(\phiout(x_2))]\\
&= \<\Om, [\phiout(x_1), \phiout(x_2)]\Om\>\cdot
  R_{\pm i}(\phiout(x_1))R_{\pm i}(\phiout(x_2))^2 R_{\pm i}(\phiout(x_1))\\
&=   \Re \<P_1x\Om, P_1x_2\Om\>\cdot R_{\pm i}(\phiout(x_1))R_{\pm i}(\phiout(x_2))^2 R_{\pm i}(\phiout(x_1)),
\end{align*}
where $\Re$ denotes the real part of the following number.
 \item For $x \in \A_{N_0}(O)$ and $y \in \F(V_{O,+})$, it holds that $[R_{\pm i}(\phiout(x)),y] = 0$.
 \end{enumerate}
\end{theorem}

We note that by Claims (\ref{lee:commutation}) and (\ref{lee:houtprod}),
the domain of $\phiout(x)$ includes $\F(V_{O,+})\houtprod$.

The restriction to $\A_{N_0}$ is essential in the original proof \cite{Buchholz77}.
The technical issue is that the set $\A_{N_0}(O)$ is covariant under Poincar\'e transformations
and dilations but not under conformal transformations.
We will extend these results to each smooth operator in a local algebra $\A(O)$. This has been expected
by Buchholz himself in the same paper \cite[P.\! 157, footnote]{Buchholz77}.

\subsubsection{Extension to general smooth operators}\label{extension}
We exploit the arguments of \cite[Chapter VI\!I\!I.7]{RSI} and \cite[Chapter X.10]{RSII}.
Let $\{A_n\}$ be a sequence of (unbounded) operators.
The following is an adaptation of \cite[Theorem X.63]{RSII} to the case of
our interest.
\begin{lemma}\label{lm:conv-general}
 Let $\{A_n\}$ be a sequence of self-adjoint operators on $\H$,
 whose domains have a dense intersection $\D$
 and suppose that their resolvents $R_{\pm i}(A_n)$ are strongly convergent, whose limits we denote by $R_\pm$
 and that for each $\xi \in \D$, $A_n\xi$ is convergent in norm, whose limit we denote
 by $A\xi$.
 Then there is a self-adjoint extension $\tilde A$ of $A$
 and $A_n$ are convergent to $\tilde A$ in the strong resolvent sense. 
\end{lemma}
\begin{proof}
 We claim that $\ker R_\pm = \{0\}$. Let $\xi \in \ker R_+$ and $\eta \in \D$.
 It is clear that $R_+^* = R_-$.
 It holds that
 \begin{eqnarray*}
  \<\xi,\eta\> &=& \<\xi, R_{-i}(A_n)(A_n-i)\eta\> \\
  &=& \<R_{+ i}(A_n)\xi, (A_n-i)\eta\> \\
  &=& \lim_n\, \<R_{+ i}(A_n)\xi, (A_n-i)\eta\> \\
  &=& \<R_+\xi, (A-i)\eta\> \\
  &=& 0.
 \end{eqnarray*}
As $\D$ is dense, $\xi = 0$. Similarly $\ker R_- = \{0\}$
and it follows that $\ran R_\pm$ are dense in $\H$ since
$R_\pm = R_\mp^*$. Then by the Trotter-Kato theorem \cite[Theorem V\!I\!I\!I.22]{RSI}
there is a self-adjoint operator $\tilde A$ and $A_n \to \tilde A$ in the strong resolvent sense.

The domain of $\tilde A$ is exactly $R_\pm\H$ and for $\xi \in \D$ it holds
that
\[
 R_\pm\cdot (A\pm i)\xi = \lim_n R_{\pm i}(A_n) (A_n\pm i)\xi= \xi,
\]
by the uniform boundedness of $R_{\pm i}(A_n)$,
hence $\xi$ is in the range of $R_\pm$ and
$\D$ is included in the domain of $\tilde A$.
\end{proof}

We do not know whether $\D$ is a core of $\tilde A$ in general.
We will prove this in the case of asymptotic fields.

Let $N_0 \ge 15$. For a smooth $x \in \A(O)$, where $O$ is a double cone, there is a sequence $x_n \in \A_{N_0}(O_n)$
such that $P_1x\Om = \lim P_1x_n\Om$ and $P_1x^*\Om = \lim P_1x_n^*\Om$ by the argument of \cite[Remark, p.155]{Buchholz77},
where $O_n$ is growing to the past of $O$.
Namely, for $n \in \NN$ one can take $\f_n(t)$ whose Fourier transform is
\[
 \tilde \f_n(\omega) = (1+(e^{-i\omega n}-1)/i\omega n)^{N_0}\cdot \tilde \f(\omega/n),
\]
where $\f$ is a test function which vanishes for $t\ge0$ and $\int dt\, \f(t) = 1$.
We define $x_n = \int dt\, \f_n(t) \Ad U(\tau(t,0))(x)$,
where $\tau$ denotes the translation.
If $x$ is self-adjoint, we may consider $x_n + x^*_n$ and assume that $x_n$ are self-adjoint as well.
It is clear that $x_n$ are contained in the union of past translations of $O$. Let $O_n$ be their localization regions.
Let $V_{O,+}$ be the future tangent of $O$, then it is the future tangent
of the finite union $O \cup O_1 \cup\cdots\cup O_n$.
By \cite[Theorem 7]{Buchholz77} cited above, all $\{\phiout(x_n)\}$ are self-adjoint.
In addition, $\F(V_{O,+})\Omega$, and accordingly $\F(V_{O,+})\houtprod$, are common cores.

\begin{lemma}\label{lm:conv-field}
 The sequence $\{\phiout(x_n)\}$ is convergent in the strong resolvent sense.
\end{lemma}
\begin{proof}
Let us denote $R_{\pm,n} = R_{\pm i}(\phiout(x_n))$.
On the subspace $\{y\Om: y \in \F(V_{O,+})\}$, which is a common core
for $\{\phiout(x_n)\}$, it holds that
$R_{\pm,n} y\Om = y R_{\pm,n}\Om$ and $y \in \F(V_{O,+})$ is bounded.
Since $\{R_{\pm,n}\}$ is uniformly bounded, it is enough to show
that $R_{\pm,n}\Om$ is convergent.

We know from \cite{Buchholz77} that $\phiout(x_n)$ acts on $\houtprod$ like
the free field. Since the problem is now reduced to the vacuum $\Om$ and
the free fields, we can restrict ourselves to $\houtprod$ and
its closure, namely the Fock space generated from $\Om$ by
the fields. Let us denote $\xi_n := P_1 x_n\Om$.
The action of the exponentiated field $e^{i\phiout(x_n)}$ on the
vacuum $\Om$ is given by $e^{i\phiout(x_n)}\Om = e^{-\frac12\<\xi_n,\xi_n\>} e^{\xi_n}$,
where we introduced a vector (cf.\! \cite{Longo08})
\[
e^{\eta} := \Om \bigoplus_k \frac{1}{\sqrt{k!}} \eta^{\otimes k}.
\]
It is easy to see that $\<e^\eta,e^\zeta\> = e^{\<\eta,\zeta\>}$.
Now it is obvious that $\eta \mapsto e^{\eta}$ is continuous.
This implies the convergence $e^{\xi_n}\to e^\xi$ when $\xi_n \to \xi$.
The exponentiated field acts by
$e^{i\phiout(x_n)}e^\eta = e^{-\frac12\<\xi_n,\xi_n\>}e^{-\<\xi_n,\eta\>}e^{\xi_n+\eta}$
and $\{e^\eta\}$ is total in the Fock space.
The whole argument applies to $t\xi_n$ for arbitrary $t\in\RR$,
hence $e^{it\phiout(x_n)}$ are strongly convergent to
$W(t\xi)$ on the Fock space (because this sequence is uniformly bounded),
where $W(\xi)$ is an operator
which acts by $W(\xi)\eta = e^{-\frac12\<\xi,\xi\>}e^{-\<\xi,\eta\>}e^{\xi+\eta}$.

Hence we obtain the convergence in the strong resolvent sense \cite[Theorem V\!I\!I\!I.21]{RSII},
in particular $R_{\pm,n}\Om$ is convergent.
\end{proof}

As seen from Theorem \ref{th:bu}(\ref{lee:houtprod}),
$\phiout(x_n)$ is convergent on $\houtprod$, hence on $\F(V_{O,+})\houtprod$.

By Lemma \ref{lm:conv-general}, there is a self-adjoint operator,
which we denote by $\Upsilon(\xi)$, such that $\Upsilon(\xi)$ is the limit
of $\{\phiout(x_n)\}$ in the strong resolvent sense.
Accordingly, $\Upsilon(\xi)$ commutes with $\F(V_{O,+})$ on its domain.
Importantly, we have shown that $\Upsilon(\xi)$ is a self-adjoint extension
of the limit of the sequence $\{\phiout(x_n)\}$ on a common domain $\F(V_{O,+})\houtprod$.
Furthermore, the action of $\Upsilon(\xi)$ is determined by $\xi$
as in Theorem \ref{th:bu}(\ref{lee:houtprod}).
This implies that $\Om$ is in the domain of $\Upsilon(\xi)^m$ for any $m \in \NN$.

\begin{lemma}\label{lm:analytic}
 Any vector $y\Om \in \F(V_{O,+})\Om$ is an analytic vector for $\Upsilon(\xi)$.
 In particular, $\F(V_{O,+})\houtprod$ is a core of $\Upsilon(\xi)$.
\end{lemma}
\begin{proof}
 We have to estimate $\Upsilon(\xi)^ky\Om$. 
The operator $\Upsilon(\xi)$ commutes with $y$ and acts on $\Om$
as the free field, hence we have
\[
\|\Upsilon(\xi)^my\Omega\| \le \|y\|\cdot \left(\sqrt{(2m)!\,2^{-m}(m!)^{-1}}\right)\cdot \|\xi\|^m.
\]
Then it is easy to see that $\sum_m \|\Upsilon(\xi)^m y\Om\| \frac{t^m}{m!}$ is finite for any $t$
and since the subspace $\F(V_{O,+})\houtprod$ of the domain is stable under $\phiout(\x)$,
by Nelson's analytic vector theorem \cite[Theorem X.39, Corollary 2]{RSII}
(the stability of the domain is important, see the reference\footnote{We thank D.\! Buchholz
for pointing out this assumption.}),
$\F(V_{O,+})\houtprod$ is a core of $\Upsilon(\xi)$.
\end{proof}

\begin{lemma}\label{lm:tangent}
 The subspace $\F(V_{O,+})\Om$ is a core of $\Upsilon(\xi)$.
\end{lemma}
\begin{proof}
 In \cite[Lemma 6]{Buchholz77}, it was shown that if $x_0 \in \A_{N_0}(O)$, $N_0 \ge 15$,
 then the domain $\D(\phiout(x_0))$
 of $\phiout(x_0)$, which is defined as the closure of the operator on $\F(V_{O,+})\Om$,
 includes $\houtprod$ and the action of $\phiout(x_0)$ on $\houtprod$ is exactly same
 as that of the free fields. Actually the only properties of $\phiout(x_0)$ used there are those
 that $\Om$ is in the domain of $\phiout(x_0)^*\phiout(x_0)$ and $\phiout(x_0)$ commute
 with $\F(V_{O,+})$, which are true also for $\Upsilon(\xi)$ as we have seen.
 
 For the reader's convenience, we review the proof of \cite[Lemma 6]{Buchholz77}.
 Let $x_0 \in \A_{N_0}(O)$.
 There is an $N$ (depending on $n$ which appears later) such that there is a sequence $\{y_k\}$
 which belongs to $\A_N(O_k)$, where $O_k \subset V_{O,+}$ (the localization region $O_k$ depends
 on $k$), $y_k\Om \to \xi_1\timeso\cdots\timeso\xi_n$ weakly and $y_k^*y_k\Om$ is uniformly
 bounded. To see that $\xi_1\timeso\cdots\timeso\xi_n$ is in the domain of $\phiout(x_0)$,
 one needs to estimate $\<\phiout(x_0)^*\eta, y_k\Om\>$ for an arbitrary vector
 $\eta \in \D(\phiout(x_0)^*)$. By using the fact that $\phiout(x_0)$ commutes with $y_k$,
 (which is also valid for $\Upsilon(\xi)$),  one obtains
 \[
  |\<\phiout(x_0)^*\eta, y_k\Om\>|^2 \le \|\eta\|^2\cdot \|\phiout(x_0)y_k\Om\|^2
  \le \|\eta\|^2\cdot \|y_k^*y_k\Om\|\cdot \|\phiout(x_0)^*\phiout(x_0)\Om\|,
 \]
 if $\phiout(x_0)\Om$ is in the domain of $\phiout(x_0)^*$ (this follows in the original proof
 from the assumption that $x_0 \in \A_{N_0}(O)$ and this is the only point where $N_0 \ge 15$ 
 is required. For $\Upsilon(\xi)$ we already know that
 that one can repeat its action on $\Om$ arbitrarily many times). This expression is uniformly bounded by
 the choice of $y_k$, hence $\<\phiout(x_0)^*\eta, \xi_1\timeso\cdots\timeso\xi_n\>$
 is bounded by $\|\eta\|$ times a constant and $\xi_1\timeso\cdots\timeso\xi_n$ belongs to $\D(\phiout(x_0))$.

 In order to get the explicit action of $\phiout(x_0)$ on $\xi_1\timeso\cdots\timeso\xi_n$ (see Theorem \ref{th:bu}),
 one takes a sequence $\{x^{(m)}\}$, where each member belongs to $\A_N(O^{(m)})$, double cones growing to
 the past of $O$ as in the construction before Lemma \ref{lm:conv-field}
 (it is not explicitly written in the original proof, but $N$ must be chosen corresponding to $2(n+1)$, see
 also \cite[Lemmas 2, 3]{Buchholz77}). In this computation, the only point is that
 $\{Px^{(m)}\Om\}$ can approximate $Px_0\Om$, which is true also for $\xi$.

 Although $\{\xi_k\}$ are not completely arbitrary since $\xi_1\timeso\cdots\timeso\xi_n$ must be the limit of $y_k\Om$,
 they form a total set in the free Fock space.
 Once one obtained the action of $\phiout(x_0)$ on a dense subspace,
 an arbitrary $n$-particle vector can be approximated in the $n$-particle subspace and the action of $\phiout(x_0)$
 is continuous there, hence by the closedness of $\phiout(x_0)$ it follows that
 any vector in $\houtprod$ is in the domain of $\phiout(x_0)$. The same argument is valid for $\Upsilon(\xi)$.

 Altogether, the closure of the restriction of $\Upsilon(\xi)$ to $\F(V_{O,+})\Om$
 includes $\F(V_{O,+})\houtprod$, hence the full domain of $\Upsilon(\xi)$ by
 Lemma \ref{lm:analytic}. This was what we had to prove.
\end{proof}

As $\phiout(x)$ is defined as the closure of the operator
$\F(V_{O,+})y \Omega\ni\eta \longmapsto yP_1x\Omega$, we can infer that
$\phiout(x) = \Upsilon(\xi)$.

\begin{theorem}\label{th:covariance}
 For any $x = x^*\in \A(O)$ smooth,
 $\phiout(x)$ is self-adjoint with a core $\F(V_{O,+})\Om$
 where $V_{O,+}$ is the future tangent of $O$.
 The sequence $\Phi^{h_T}(x)$ is convergent to $\phiout(x)$ in the strong resolvent sense.
\end{theorem}
\begin{proof}
 By definition, $\phiout(x)$ is the closure of the operator $y\Om \mapsto yP_1x\Om$ on $\F(V_{O,+})\Om$.
 But since $\Upsilon(\xi) (= \Upsilon(P_1x\Om))$ is self-adjoint and $\F(V_{O,+})\Om$ is its core,
 it follows that $\Upsilon(\xi) = \phiout(x)$, as their actions coincide on their cores.
 
 As for the convergence, we follow the proof of \cite[Theorem 9]{Buchholz77}.
 We know that $\F(V_{O,+})\Omega$ is a core for $\phiout(x)$
 and it is self-adjoint. For $y \in \F(V_{O,+})$,
 \[
  \underset{T\to\infty}\slim (\Phi^{h_T}(x) + \lambda)^{-1}(\phiout(x) + \lambda)y\Omega
  = \underset{T\to\infty}\slim (\Phi^{h_T}(x) + \lambda)^{-1}(\Phi^{h_T}(x) + \lambda)y\Omega = y\Omega
 \]
by the uniform boundedness of $(\Phi^{h_T}(x) + \lambda)^{-1}$ for a fixed $\lambda \notin \RR$.
By the self-adjointness of $\phiout(x)$, $\{(\phiout(x) + \lambda)y\Omega, y\in \F(V_{O,+})\}$ is
dense in $\H$ and we obtain the convergence in the strong resolvent sense,
again by the uniform boundedness of the sequence.
\end{proof}

\begin{lemma}\label{lm:core}
 Let $(\A, U, \Om)$ be a conformal net.
 For $x = x^* \in \A(O)$ smooth, there is a $O_+$ whose closure is contained
 in the future tangent $V_{O,+}$ of $O$ such that $\A(\O_+)\Om$ is a core for $\phiout(x)$.
\end{lemma}
\begin{proof}
 We work on the extension of $\A$ on $\cyl$ and the lift of $U$ to $\gconf$.
 
 Recall that $V_{O,+}$ is
 a translation of the future lightcone, then there is a region $D$ in $\cyl$ such
 that the inclusion $V_{O,+} \subset D$ is conformally equivalent to
 $O_+ \subset V_+$, where $O_+$ is a double cone whose past apex is the point of origin.
 Then the conformal transformations associated to $V_+$, dilations, shrink $O_+$.
 Accordingly the conformal transformations associated to $D$ shrink $V_{O,+}$ to
 double cones whose past apex is the apex of $V_{O,+}$ (see Figure \ref{fig:duality}).
 In this situation, such a transformation shrinks also $O$.
 
 Let $g$ be a conformal transformation as in the previous paragraph.
 Now the operator $\phiout(\Ad U(g)(x))$ has a core $\F(V_{O,+})\Om$
 and $\Ad U(g)(\phiout(x))$ has a core $U(g)\F(V_{O,+})\Om = \F(gV_{O,+})\Om$,
 where $\F(gV_{O,+})$ is analogously defined as $\F(V_{O,+})$.
 Their actions coincide on $\F(gV_{O,+})\Om$, namely for $y\in \F(gV_{O,+})$
 they give $y\Om \mapsto yU(g)P_1x\Om = yP_1U(g)x\Om$ (the conformal group
 preserves $P_1\H$ from the classification of unitary positive-energy representations,
 Section \ref{representation}).
 The operator $\phiout(\Ad U(g)(x))$ is a self-adjoint extension of $\Ad U(g)(\phiout(x))$
 which is also self-adjoint, hence they must coincide.
 
 In the discussion above, the domain of $\phiout(\Ad U(g)(x))$
 naturally includes $\A(gV_{O,+})\Omega$ (note that $\A(gV_{O,+})$
 is a von Neumann algebra).
 Reversing the argument, for any $x\in\A(O)$
 there is a sufficiently large double cone $O_+$ in $V_{O,+}$,
 whose past apex is the future apex of $O$,
 such that $\A(O_+)\Omega$ is a core of $\phiout(x)$.

 Until now in this proof and in Theorem \ref{th:covariance},
 regarding the localization,
 we used only the assumption that $x$ is localized in $O$,
 a double cone in the past tangent of $V_{O,+}$.
 By considering $\Ad U(\t(-a))(x)$ which is localized in $O-a$ for a future-timelike vector $a$
 and translating everything by $a$ after the argument,
 we see actually that $\A(O_+ + a)\Omega$ is a core of $\phiout(x)$.
 In other words, if $x$ is localized in a double cone, then there is another double cone
 in the future tangent, separated by a nontrivial timelike vector,
 whose local operators can generate a core for $\phiout(x)$.
\end{proof}

\begin{corollary}\label{cr:covariance}
 Let $(\A, U, \Om)$ be a conformal net.
 For $x=x^*\in\A(O)$ smooth and $g\in\gconf$ sufficiently near to the unit element
 such that $gO$ is still a double cone in the Minkowski space $M$, it holds that
 $\Ad U(g)(\phiout(x)) = \phiout(\Ad U(g)(x))$.
\end{corollary}
\begin{proof}
 We may assume that $x$ is localized in $\check O$, whose closure is still in
 $O$. Let $O_+ + a$ be a double cone in $V_{O,+}$ separated from the future apex of $O$
 such that $\A(O_+ + a)\Om$ is a core for $\phiout(x)$
 (Lemma \ref{lm:core}).
 If $g \in \gconf$ is sufficiently near to the unit, we may assume the following:
 \begin{itemize}
  \item $g\check O \subset O$,
  \item $gO$ and $g(O_+ + a)$ are included in $\RR^4$,
  \item there is a double cone $\widehat O_+$ which include $(O_+ + a) \cup g(O_+ + a)$
 such that $\widehat O_+$ and $g^{-1}\widehat O_+$ are in the future tangent $V_{O,+}$ of $O$.
 \end{itemize}
 The set $\A(\widehat O_+)\Omega$  is a core of $\Ad U(g)(\phiout(x))$
 and $\phiout(\Ad U(g)(x))$. But their actions on $\Om$ coincide and they
 commute with $\A(\widehat O_+)$, hence the operators must coincide.
 This concludes the desired local covariance of $\phiout(x)$ with respect to
 $U$.
\end{proof}

We can now define the outgoing free field net by
\[
\A^\tout(O) := \{R_\lambda(\phiout(x)): x = x^*\in\A(O) \mbox{ smooth}, \lambda \notin \RR\}''.
\]
By Corollary \ref{cr:covariance}, this net $\A^\tout$ is covariant
with respect to the unitary representation $U$ for the original net $\A$.
The vacuum $\Om$ is in general not cyclic for $\A^\tout$.

This free field net can be defined for any given net which contains massless particles.
We will show that it is a subnet for a given conformal net,
namely $\A^\tout(O) \subset \A(O)$.

\section{A proof under global conformal invariance}\label{gci}
In this Section we show that a globally conformal net (defined below)
contains the second quantization (free) net if it has nontrivial massless particle spectrum.
Of course these two assumptions are very strong.
We can actually drop global conformal invariance as we will see in Section \ref{general} but here
we present a simpler proof in order to clarify the involved ideas.
This result should thus be considered as a simplification in operator-algebraic
formulation of \cite{Baumann82} with an additional
assumption, the global conformal invariance (GCI).
It is a strong property, under which there are indications that
the stress-energy tensor is the same as that of the free field \cite{Stanev13}.

A conformal net $(\A,U,\Om)$ is said to be {\bf globally conformal} if
the extension to $\cmink$ (the compactified Minkowski space, see Section \ref{conformalgroup}) already admits a
global action of $\gconf$ (cf.\! \cite{NT01, NT04}, where GCI is defined
in terms of Wightman functions). Namely,
the action of $\gconf$ factors through the action of $\conf$.
For example, the massless free fields with odd integer helicity are globally conformal,
while other free fields are not \cite[Corollary 3.12]{Hislop88}.

In this case, any two operators $x,y$ localized in timelike-separated
regions commute. Indeed, any pair of timelike-separated regions can be
brought into spacelike-separated regions by an action of $\conf$.

The first consequence of GCI is the following.
\begin{proposition}\label{pr:gci-duality}
 For a net $\A$ with GCI, it holds that $\A(V_+) = \A(V_-)'$,
 where $V_\pm$ are the future and past lightcones.
\end{proposition}
\begin{proof}
 As remarked above, it holds that $\A(V_+) \subset \A(V_-)'$ by GCI.
 The modular group for $\A(V_-)$ with respect to $\Om$ is the dilation \cite{BGL93}
 (see Section \ref{bisognano-wichmann}),
 thus the modular group for $\A(V_-)'$ with respect to $\Om$ is again dilation
 (up to a reparametrization).  It is clear that $\A(V_+)$ is invariant under dilation.
 
 Let us recall the simple variant of Takesaki's theorem \cite[Theorem IX.4.2]{TakesakiII}.
 Assume that $\N \subset \M$ is an inclusion of von Neumann algebras, $\Om$ is a cyclic separating
 vector for $\M$ and the modular group $\Ad \Delta^{it}$ for $\M$ with respect to $\Om$
 preserves $\N$. Then there is a conditional expectation $E:\M\to\N$ which preserves the
 state $\<\Om, \cdot\,\Om\>$ and this is implemented by the projection $P$ onto the subspace
 $\overline{\N\Om}$: $E(x)\Om = Px\Om$.
 In particular, $E(x) = x$ if and only if $x\in\N$.
 
 In our situation, from Takesaki's theorem it follows that $\A(V_+) = \A(V_-)'$
 because $\Om$ is cyclic for the both algebras
 by Reeh-Schlieder property (cf.\! \cite[Appendix A]{Tanimoto12-1}).
 Therefore the projection above is trivial and the two von Neumann algebras must coincide. 
\end{proof}

\begin{lemma}\label{gci:subnet}
 For a net $\A$ with GCI,
 the outgoing free field net $\A^\tout$ is a subnet of $\A$.
\end{lemma}
\begin{proof}
 Let $O \subset V_-$ and $O_+ \subset V_+$. In particular,
 $O_+$ is in the future tangent of $O$.
 By the construction of asymptotic fields, $\Phi^{h_T}(x)$ is
 in the spacelike complement of $\A(O_+)$ if $x\in \A(O)$, hence we have
 $R_\lambda(\phiout(x)) \in \A(V_+)'$ by the convergence in the strong resolvent sense
 and by Proposition \ref{pr:gci-duality}
 this is equal to $\A(V_-)$. This implies that $\A^\tout(V_-) \subset \A(V_-)$.
 
 By conformal covariance with respect to the same representation $U$ (see the end of
 Section \ref{extension}), with the conformal group $\conf$ which takes $V_-$ to any double cone $O$,
 we obtain $\A^\tout(O) \subset \A(O)$.
\end{proof}

We summarize the result.
\begin{theorem}\label{gci:main}
 Let $(\A, U, \Om)$ be a globally conformal net and assume that
 the massless particle spectrum of $U$ is nontrivial.
 Then there is a subnet $\A^\tout$ of $\A$, which is isomorphic to
 the free field net associated to the massless representation.
 The free subnet $\A^\tout$ generates the whole massless particle spectrum of $U$.
\end{theorem}
\begin{proof}
 Almost all statements have been proved above.
 The whole massless particle spectrum of $U$ is generated by
 $\A^\tout$ since $\{P_1x\Om: x\in \A(O)\}$ is dense in $P_1\H$ by
 the Reeh-Schlieder property of $\A$ and we only have to consider
 the asymptotic fields for self-adjoint elements $x_+ = (x+x^*)/2$ and $x_- = (x-x^*)/2i$.
 The exponentiated fields $e^{i\phiout(x_\pm)}$ are localized in
 $\A^\tout(O)$ and the one-particle vectors are obtained by
$\frac{d}{dt}e^{it\phiout(x_\pm)}\Om$.
\end{proof}

One can analogously define $\A^\tin$ by taking the limit $T \to -\infty$.
Now that we know that the net $\A$ includes
a free field subnet, it follows that $\A^\tout = \A^\tin$
because we can choose local operators $x$ which creates one-particle states
from the free subnet. For the free field net, the asymptotic field net
is of course itself, so we obtain $\A^\tout = \A^\tin$.
Accordingly, although one can define S-matrix on the subspace generated by
$\A^\tout = \A^\tin$, roughly as the difference between
$\xi_1\timeso\cdots\timeso\xi_n$ and $\xi_1\timesi\cdots\timesi\xi_n$.
(see \cite{BF77}, and \cite{Buchholz75} for its two-dimensional variant), it is trivial.

\section{A general proof}\label{general}
Finally let us prove the existence of a free subnet under conformal invariance
but not necessarily under global conformal invariance. If a net is not globally conformal,
it does not necessarily hold that $\A(V_+)' = \A(V_-)$ and our previous argument
does not work. Instead, here we use directed asymptotic fields defined below.
As already suggested by Buchholz himself \cite[Section 4]{Buchholz82},
Theorem \ref{th:bu} can be extended for asymptotic fields with a function $f$
which specifies a direction in which a local observable proceeds asymptotically.
Such a directed asymptotic field still has a certain local property
and we can construct subnet.

\subsection{Directed asymptotic fields}\label{directed}
For a smooth function $f$ on the unit sphere $S^2$ such that
$f(\mathbf n) \ge 0$ and $\int_{S^2} d\o(\mathbf n)\; f(\mathbf n) = 1$, we define
\[
 \Phi^t_f(x) := -2t\int_{S^2} d\o(\mathbf n)\; f(\mathbf n) \partial_0 x(t,t \mathbf n), \;\;\;\Phi^{h_T}_f(x) = \int_{\RR} dt\; h_T(t)\Phi^t_f(x).
\]
where notations are as in Section \ref{convergence}.
In \cite{Buchholz77} the case where $f=1$ has been worked out and
it has been suggested in \cite{Buchholz82} that the whole theory
works for a general $f$. As we need certain extended results,
let us discuss the proofs and how they should be modified when
$f$ is nontrivial.

First, we explain the claim \cite[Equation (4.3)]{Buchholz82}:
\[
 \underset{T\to \infty}\slim \Phi^{h_T}_f(x)\Omega = P_1 f\left(\frac{\mathbf P}{|\mathbf P|}\right)x\Omega,
\]
where $\mathbf P$ is the 3-momentum operator of the given representation $U$ of the net
(see Section \ref{convergence})
and $f\left(\frac{\mathbf P}{|\mathbf P|}\right)$ is defined by functional calculus.
This follows from the mean ergodic theorem analogously as in \cite[Section 2]{Buchholz77}.
Indeed, this time we have
\[
 \Phi^t_f(x)\Omega = -\frac{it}{2\pi}\int dE_P \int_0^\pi \sin \theta d\theta \int_0^{2\pi}d\f\;
 f(\theta,\f) e^{it(H - \mathbf{n}\cdot \mathbf P)}H(x\Omega)_P
\]
where $P = (H, \mathbf P), \mathbf n = (\sin \theta \cos \f, \sin \theta \sin \f, \cos \f)$
and the integral is about $\mathbf n$ (on the unit sphere)
and the joint spectral decomposition
with respect to $P$ and accordingly $(x\Omega)_P$ is the $P$-component with respect to it.
Since the support of $P$ is included in the closed positive lightcone
$\overline V_+$, the $t$-dependent phase vanishes $e^{it(H - \mathbf n\cdot \mathbf P)}$ only on
the surface of the cone $H = |\mathbf P|$. Instead, on this surface
the integral with respect to $\theta,\f$ gives
$\frac{2\pi}{-it|\mathbf P|}f\left(\frac{\mathbf P}{|\mathbf P|}\right)e^{it(H -|\mathbf P|)}$ with additional terms
which tend to zero when the limit in the mean ergodic theorem is taken
(this can be explicitly demonstrated by considering a function $f$ which is
$z$-rotation symmetric. A general function can be approximated by sums of such
functions with different axis of symmetry in $L^1$-norm). Hence we obtain the formula above.

Only in this paragraph, the propositions and sections refer to those in \cite{Buchholz77}.
Now, Lemma 1 can be modified straightforwardly. Lemma 2 is the main technical ingredient
and has been proved in the Appendix. Now, among the statements in the Appendix,
the only one in which the spherical integral matters is the Lemma, in which commutators
of spherically smeared operators are estimated. Here the only property essentially used
in the estimate is locality of operators and the integrand gets bounded by norm.
This means, if one has to smear the integrand with $f$, it changes the weight
of localization. However, as the integrand is bounded by norm and no other
technique is required, one can simply bound $f$ by a constant in order to adapt
the proof. By this bound, the estimate gets simply multiplied by a constant
depending on $f$. This does not affect the rest of the arguments at all.
Indeed, this Lemma is used later in Corollary, and indirectly in Proposition II,
where the overall constant is unimportant.
Finally, Lemma 2 is proved in Section d) and the overall constant in the estimate
does not play any role, hence we obtain the modified Lemma 2.
In the rest of the paper, the spherical integral appear only through
the correspondence from $x$ to $P_1 f\left(\frac{\mathbf P}{|\mathbf P|}\right)x\Omega$.
Accordingly, one can modify all the propositions of the paper.

Thereafter one can repeat our argument in order to extend the results from
$\A_{N_0}(O)$ to $\A(O)$. In summary, we obtain the following.

\begin{theorem}\label{th:directed}
Let $x = x^*, x_1 = x_1^*, x_2 = x_2^*$ be smooth elements (with respect to
$\gconf$) of $\A(O)$,
$O$ be a double cone and $f, f_1, f_2$ be smooth functions on $S^2$.
\begin{enumerate}[{(}1{)}]
\item\label{dir:commutation} For arbitrary $y \in \A(O_+)$, where $O_+ \subset V_{O,+}$ is bounded,
 $y\cdot \D(\phiout_f(x))\subset \D(\phiout_f(x))$ and
 one has $[\phiout_f(x), y] = 0$ on $\D(\phiout_f(x))$.
 \item\label{dir:core} The operator $\phiout_f(x)$ is self-adjoint and depends only on
 $P_1 f\left(\frac{\mathbf P}{|\mathbf P|}\right)x\Om$. The subspace $\F(V_{O,+})\Om$ is a core of $\phiout_f(x)$.
 \item\label{dir:converngece} The sequence $\Phi^{h_T}_f(x)$ is convergent to $\phiout_f(x)$ in the
 strong resolvent sense.
 \item\label{dir:houtprod} The domain $\D(\phiout_f(x))$ includes the set $\houtprod$ of
 all product states
 $\xi_1\timeso\xi_2\timeso \cdots \timeso \xi_n$ and its action is
 \[
  \phiout_f(x)\cdot \xi_1\timeso\xi_2\timeso \cdots \timeso \xi_n =
  \xi\timeso\xi_1\timeso\xi_2\timeso \cdots \timeso \xi_n +
  \sum_{k=1}^n \<\xi,\xi_k\> \xi_1\timeso\cdots \check{\xi}_k\cdots \timeso\xi_n,
 \]
 where
 $\xi = P_1f\left(\frac{\mathbf P}{|\mathbf P|}\right)x\Om = P_1f\left(\frac{\mathbf P}{|\mathbf P|}\right)x^*\Om$.
 \item For the resolvent $R_{\pm i}(y)= (y\pm i)^{-1}$ of $y$, it holds that
\begin{align*}
&[R_{\pm i}(\phiout_{f_1}(x_1)), R_{\pm i}(\phiout_{f_2}(x_2))]\\
&= \<\Om, [\phiout_{f_1}(x_1), \phiout_{f_2}(x_2)]\Om\>\cdot
  R_{\pm i}(\phiout_{f_1}(x_1))R_{\pm i}(\phiout_{f_2}(x_2))^2 R_{\pm i}(\phiout_{f_1}(x_1))\\
&=   \Re \left\<P_1f_1\left(\frac{\mathbf P}{|\mathbf P|}\right)x_1\Om, P_1f_2\left(\frac{\mathbf P}{|\mathbf P|}\right)x_2\Om\right\>
\cdot R_{\pm i}(\phiout_{f_1}(x_1))R_{\pm i}(\phiout_{f_2}(x_2))^2 R_{\pm i}(\phiout_{f_1}(x_1)).
\end{align*}
 \item For $x \in \A(O)$ and $y \in \F(V_{O,+})$, it holds that $[R_{\pm i}(\phiout_f(x)),y] = 0$.

 \end{enumerate}
\end{theorem}
Other propositions in \cite[Section 4]{Buchholz77} can be appropriately
modified but we state here only what we need.

\subsection{Conformal free subnet}
Let $\A$ be a conformal net with massless particles. We consider the standard double cone $O_1$.
The following is an easy geometric observation (c.f. \cite[P.60]{Buchholz82}).
\begin{lemma}\label{lm:spacelike}
 For a double cone $O$ which is sufficiently spacelike separated from $O_1$,
 there is a compact set $\Sigma$  in $S^2$ such that
 $\{a+(t,t\mathbf n): a\in O, \mathbf n\in \Sigma, t \mbox{ sufficiently large}\}$
 is spacelike separated from $O_1$.
\end{lemma}
Let us explain what ``sufficiently separated'' means.
First, we consider for simplicity the point of origin and a spacelike vector $v$.
We may assume that
$v = (v_0,0,0,v_3)$, where $|v_0| < v_3$. The vectors in question are of the form
\[
\{(v_0+t, t\sin\theta\,\cos\phi, t\sin\theta\, \sin\phi, v_3+t\cos\theta), t\ge 0\}.
\]
As one can check easily, these are spacelike for sufficiently large $t$
if $\cos\theta > \frac{v_0}{v_3}$.
In general, even if $O$ and $O_1$ are open regions, if the difference $O_1 - O$
is almost in one direction, then the above arguments works.

From this, we see that certain directed asymptotic fields still have
certain locality.
\begin{lemma}\label{lm:directed}
 For $x \in \A(O)$ where $O \perp O_1$ (spacelike separated) and a smooth function $f$
 such that $O$ and the support of $f$ satisfy the situation of Lemma \ref{lm:spacelike},
 $\phiout_f(x)$ is affiliated to $\A(O_1)' = \A(O_1^\cc)$.
\end{lemma}
\begin{proof}
 This follows immediately from the localization of approximants
 $\Phi^{h_T}_f(x)$ and their convergence to $\phiout_f(x)$ in the strong resolvent sense.
\end{proof}

We construct a subnet of $\A$ as follows. First, consider the following:
\begin{align*}
\A^\tdir(O_1^\cc) :=& \{ \Ad U(g)(R_\lambda (\phiout_f(x))):
\Im \lambda \neq 0, g \in \gconf(O_1),\\
& x \in \A(O), O\perp O_1, f \mbox{ as Lemma \ref{lm:directed}}\}'',
\end{align*}
where $\gconf(O_1)$ is the stabilizer group of $O_1$ in $\gconf$.
This is clearly a subalgebra of $\A(O_1^\cc) = \A(O_1)'$.
For any other double cone $O$
in the global space $\cyl$, we can find $g \in \gconf$ such that $O = gO_1^\cc$.
With this $g$, we define $\A^\tdir(O) = \Ad U(g) (\A^\tdir(O_1^\cc))$.
This is well-defined, because in the definition of $\A^\tdir(O_1^\cc)$ above
$g$ runs in the stability group $\gconf(O_1)$.

\begin{lemma}\label{lm:subnetdir}
The family $\{\A^\tdir(O)\}$ is a conformal subnet of $\A$
and generates $\hout$ from the vacuum $\Om$.
\end{lemma}
\begin{proof}
Covariance of $\A^\tdir$ holds by definition (and well-definedness).
$\A^\tdir(O)$ is a subalgebra of $\A(O)$, hence locality follows.
Positivity of energy and the properties of vacuum are inherited from those of
$U$ and $\Om$.

Note that the closed subspace $\hout$ = $\overline{\houtprod}$ is invariant under
$U(g)$. Indeed, we know already that $\A^\tout$ is a net whose restriction to the
Minkowski space $M$ generates the subspace $\hout$.
Any local algebra $\A^\tout(O)$, where $O$ is a double cone in $M$, produces
a dense subspace of $\hout$ from $\Om$ and if $g$ is in a small neighborhood of
the unit element of $\gconf$, then $\A^\tout(gO)$ is again a local algebra in $M$
and generate another dense subspace of $\hout$, thus
$\hout$ is invariant under such $U(g)$. A general element $g$ can be reached
as a finite product of such elements, and the invariance follows.

For $O \perp O_1$, the fields $\phiout_f(x), x\in \A(O)$ can generate
$P_1 \chi_\Sigma\left(\frac{\mathbf P}{|\mathbf P|}\right)\H$
where $\Sigma$ is the compact set in Lemma \ref{lm:spacelike}
and $\chi_\Sigma$ denotes the characteristic function of $\Sigma$.
One can patch such $\Sigma$ to see that the whole one particle space
is spanned by $\phiout_f(x)$ which are affiliated to $\A^\tdir(O_1^\cc)$.
Since the second quantization structure is the same, $\overline{\A^\tdir(O_1^\cc)\Om}$
includes the whole free Fock space $\hout$. As $\hout$ is invariant under $U(g)$,
by the construction of $\A^\tdir(O_1^\cc)$,
$\hout$ is the Hilbert subspace generated by $\A^\tdir(O_1^\cc)$ from $\Om$.
Then the same holds for an arbitrary double cone by the covariance of $\A^\tdir$ and
the invariance of $\hout$.
This is Reeh-Schlieder property of $\A^\tdir$ (as a subnet).

Now we consider the isotony of $\A^\tdir$.
The modular group of $\A(O)$ acts geometrically and $\A^\tdir(O)$ is invariant
under that by construction. By Takesaki's theorem, there is a conditional
expectation $E^\tdir$ from $\A(O)$ to $\A^\tdir(O)$ implemented by the projection $P^\tout$
onto $\hout$. It is immediate that this defines a coherent family of conditional expectations
in the sense that $E^\tdir$ does not depend on $O$, because it is implemented by the same
projection $P^\tout$. With this, the isotony of $\A^\tdir$ follows from the isotony of $\A$.
\end{proof}

\begin{proposition}
 Two nets $\A^\tdir(O)$ and $\A^\tout(O)$ coincide, the latter being
 defined in Section \ref{extension}.
\end{proposition}
\begin{proof}
If $x\in \A(O)$ and $y\in\A(O_1)$, where $O\perp O_1$ and $f$ is chosen
for the pair $O,O_1$ as in Lemma \ref{lm:spacelike},
then $\phiout_f(x)$ and $\phiout(y)$, or their resolvents,
commute by the techniques of
Jost-Lehmann-Dyson representation as in \cite[Section 4]{Landau70}\cite[Theorem 9]{Buchholz77}.
We know that $\A^\tout$ is covariant with respect to $U$.
Especially, $\A^\tout(O_1)$ is invariant under $\Ad U(g)$ where
$g \in \gconf(O_1)$.
By definition of $\A^\tdir$, the two nets $\A^\tdir$ and $\A^\tout$ are
relatively local.

We saw also that they generate the same Hilbert subspace $\hout$
in Lemma \ref{lm:subnetdir}.
Both nets $\A^\tout$, $\A^\tdir$ are conformal with respect to $U$, relatively
local and span the same Hilbert subspace.
By the standard application of Takesaki's theorem as in Proposition \ref{pr:gci-duality},
these local algebras coincide.
\end{proof}

This concludes our construction. Any conformal net, global or not, contains
a free subnet $\A^\tout = \A^\tdir$ which generates the massless particle spectrum.

\subsubsection*{Decoupling of the free field subnet}\label{decoupling}
The next Proposition works with Haag dual (for double cones in $M$)
nets with covariance with respect to the Poincar\'e group.
A net has {\bf split property} if for each pair $O_1 \subset O_2$ such that
$\overline O_1 \subset O_2$, there is a type I factor $\R$ such that
$\A(O_1) \subset \R \subset A(O_2)$.
A {\bf DHR sector} of the net $\A$ is the equivalence class of a representation $\pi$
of the global $C^*$-algebra $\overline{\bigcup_{O}\A(O)}^{\|\cdot\|}$
where $O$ are double cones under certain conditions \cite{Haag96}. Among others,
the most important one is that there is a double cone $O$ such that
the restriction of $\pi$ to $\overline{\bigcup_{O'\perp O}\A(O')}^{\|\cdot\|}$
($\perp$ denotes the spacelike separation) is unitarily equivalent to
the identity representation (the vacuum representation).

\begin{proposition}\label{pr:decoupling}
Let $\A$ be a Haag dual subnet of a Haag dual net $\F$ on a separable Hilbert space
and assume that $\A$ has split property and has no
nontrivial irreducible DHR sector
(if $\A \subset \F$ is an inclusion of conformal nets, we have the Haag duality
on $\cyl$ and we do not need the Haag duality on $M$).
Then $\F$ decouples, namely $\F(O) = \tilde\pi_0(\A(O))\otimes \C_0(O)$
where $\C(O) = \A(O)'\cap \F(O)$ is the coset net, $\C_0$
is the irreducible vacuum representation of $\C$ and $\tilde\pi_0$
is the vacuum representation of $\A$ (the restriction of $\A$ to its cyclic subspace).
\end{proposition}
\begin{proof}
 The argument here is essentially contained in the proof of \cite[Theorem 3.4]{CC01}
 and has been suggested to apply to globally conformal nets in \cite{BNRT07}.
 
 The representation of $\A$ on the vacuum Hilbert space of $\F$ is a DHR representation
 of $\A$ \cite[Lemma 3.1]{CC01} (this can be proved under split property of $\A$ only,
 from which it follows that local algebras are properly infinite, and separability
 of the Hilbert space), hence by split property it is the direct integral of irreducible
 representations (see \cite[Proposition 56]{KLM01}, which is written for nets on $S^1$
 but the arguments apply to nets on $M$), and by assumption it is the direct sum of copies of the vacuum
 representation.
 Hence on the Hilbert space of $\F$, an element $x \in \A(O)$ is of the form
 $\tilde\pi_0(x)\otimes \CC\1$ with an appropriate decomposition $\H = \H_\A \otimes \K$.
 Since $\A$ is Haag dual on its vacuum representation $\tilde\pi_0$,
 we have $\A(O') = \tilde\pi_0(\A(O'))\otimes \CC\1 = \tilde\pi_0(\A(O))'\otimes\CC\1$.
 By the relative locality of $\F$ to $\A$,
 we have $\F(O) \subset \A(O')' = \tilde \pi_0(\A(O))\otimes \B(\K)$.
 Now we have an inclusion
 \[
  \A(O) = \tilde\pi_0(\A(O))\otimes \CC\1\subset \F(O) \subset \tilde\pi_0(\A(O))\otimes \B(\K).
 \]
 This relation holds also for a wedge $W$,
 \[
  \A(W) = \tilde\pi_0(\A(W))\otimes \CC\1\subset \F(W) \subset \tilde\pi_0(\A(W))\otimes \B(\K)
 \]
 but the wedge algebra $\tilde\pi_0(\A(W))$ in the vacuum representation is a factor
 \cite[1.10.9 Corollary]{Baumgaertel}. Now by \cite[Theorem A]{GK96},
 there is $\C_0(W) \subset \B(\K)$ such that $\F(W) = \tilde\pi_0(\A(W))\otimes \C_0(W)$.
 It is clear that $\F(W) = \A(W) \vee \C(W)$, where $\C(W) = \F(W) \cap \A(W)'$

 By Haag duality of the both nets $\F$ and $\A$, we have
\[
 \F(O) = \bigcap_{O\subset W} \F(W) 
 = \bigcap_{O\subset W}\tilde\pi_0(\A(W))\otimes \C_0(W) 
 = \tilde\pi_0(\A(O))\otimes \bigcap_{O\subset W} \C_0(W).
\]
 By defining $\C(O) := \F(O) \cap \A(O)' = \CC\1\otimes \bigcap_{O\subset W} \C_0(W)$
 and $\C_0(O) = \bigcap_{O\subset W} \C_0(W)$,
 we obtain $\F(O) = \tilde\pi_0(\A(O)) \otimes \C_0(O) = \A(O)\vee \C(O)$.

 If $\A\subset\F$ is an inclusion of conformal nets, we can directly argue with double cones $O$.
 Each $\A(O)$ is a factor, the modular group acts geometrically and Haag duality holds
 on $\cyl$ (one should simply transplant the duality argument to $\cyl$) \cite{BGL93}.
\end{proof}

\begin{corollary}
  Let $(\A, U, \Om)$ be a conformal net and assume that
 the massless particle subspace $P_1\H$ of $U$ consists only of the scalar representation
 with finite multiplicity.
 Then the free subnet $\A^\tout$ decouples in $\A$, namely
 $\A(O) = \A^\tout(O)\vee \C(O)$, where $\C(O) := \A(O) \cap \A^\tout(O)'$
 is the coset subnet.
\end{corollary}
\begin{proof}
  The scalar free field net has no nontrivial DHR sector \cite{Araki64, Driessler79}
  and has split property \cite{BJ87, BW86}. These properties are inherited by
  any finite tensor product.
  Thus the claim follows from Proposition \ref{pr:decoupling}.
\end{proof}

\section{Open problems}\label{open}
We have shown that massless particles in a conformal net are free.
However, massless representations are only one of the families of
the irreducible representations of the conformal group.
Unfortunately, at the moment the scattering theory, which extracts free fields,
is not applicable to the rest of the family.
It would be interesting if one could extract other fields by a different device.
This would not be very easy because in general they are expected to be interacting
(e.g.\! the super Yang-Mills theory \cite{Mandelstam83}).

As for decoupling, it relies on the split property and the absence of DHR sector
of the scalar free field. As the proofs in the scalar case are based on the
arguments in the one particle space and the second quantization, we
expect that similar results should hold for
each massless finite-helicity representation of the conformal group.

Another interesting question is whether it is possible to prove conformal covariance
from scale invariance (under certain additional conditions).
Some results have been obtained in this direction
\cite{Nakayama13, DKSS13}. An operator-algebraic proof is unknown
(if we do not assume asymptotic completeness, c.f.\! \cite{Tanimoto12-1}).

By comparing with the result that any massless asymptotically complete
model in two dimensions can be obtained
by ``twisting'' a tensor product net \cite[Section 3]{Tanimoto12-2} \cite[Proposition 2.2]{BT12},
one may wonder whether such a structure is available in four dimensions, too.
This is not straightforward, because wedges are not suited for
the scattering theory in four dimensions. Neither are lightcones, because
the intersection of the shifted future and past lightcones does not give back
the algebra for a double cone even in the free field net \cite{HL82}.
Related to this issue is whether the S-matrix is a complete invariant of a net
under asymptotic completeness. This is open also for massive theories, although
partial results are available \cite{BBS01, Mund12}.

\subsubsection*{Acknowledgement}
I am grateful to Detlev Buchholz and Karl-Henning Rehren for pointing out serious 
technical issues in early versions of this paper.
I thank Marcel Bischoff, Wojciech Dybalski and Nikolay Nikolov for interesting discussions,
Luca Giorgetti and Vincenzo Morinelli for useful comments and
the referees of {\it Forum of Mathematics, Sigma} for careful reading of the manuscript and
suggestions.
I appreciate the support by Hausdorff Institut f\"ur Mathematik, where a part of this work has been done.

This work was supported by Grant-in-Aid for JSPS fellows 25-205.

\def\cprime{$'$}


\begin{thebibliography}{10}

\bibitem{Araki64}
H.\! Araki.
\newblock von {N}eumann algebras of local observables for free scalar field.
\newblock {\em J. Mathematical Phys.}, 5:1--13, 1964.

\bibitem{BNRT07}
B.\! Bakalov, N.\! M.\! Nikolov, K.-H.\! Rehren and I.\! Todorov.
\newblock Unitary positive-energy representations of scalar bilocal quantum
  fields.
\newblock {\em Comm. Math. Phys.}, 271(1):223--246, 2007.

\bibitem{Baumann82}
K.\! Baumann.
\newblock All massless, scalar fields with trivial {$S$}-matrix are
  {W}ick-polynomials.
\newblock {\em Comm. Math. Phys.}, 86(2):247--256, 1982.

\bibitem{Baumgaertel}
H.\! Baumg{\"a}rtel.
\newblock {\em Operator algebraic methods in quantum field theory}.
\newblock Akademie Verlag, Berlin, 1995.

\bibitem{BT12}
M.\! Bischoff and Y.\! Tanimoto.
\newblock Construction of {W}edge-{L}ocal {N}ets of {O}bservables through
  {L}ongo-{W}itten {E}ndomorphisms. {II}.
\newblock {\em Comm. Math. Phys.}, 317(3):667--695, 2013.

\bibitem{BBS01}
H.-J.\! Borchers, D.\! Buchholz and B.\! Schroer.
\newblock Polarization-free generators and the {$S$}-matrix.
\newblock {\em Comm. Math. Phys.}, 219(1):125--140, 2001.

\bibitem{BGL93}
R.\! Brunetti, D.\! Guido and R.\! Longo.
\newblock Modular structure and duality in conformal quantum field theory.
\newblock {\em Comm. Math. Phys.}, 156(1):201--219, 1993.

\bibitem{BJ87}
D.\! Buchholz and P.\! Jacobi.
\newblock On the nuclearity condition for massless fields.
\newblock {\em Lett. Math. Phys.}, 13(4):313--323, 1987.

\bibitem{Buchholz75}
D.\! Buchholz.
\newblock Collision theory for waves in two dimensions and a characterization of models with trivial S-matrix.
\newblock {\em Comm. Math. Phys.}, 45(1):1--8, 1975.

\bibitem{Buchholz77}
D.\! Buchholz.
\newblock Collision theory for massless bosons.
\newblock {\em Comm. Math. Phys.}, 52(2):147--173, 1977.

\bibitem{Buchholz82}
D.\! Buchholz.
\newblock The physical state space of quantum electrodynamics.
\newblock {\em Comm. Math. Phys.}, 85(1):49--71, 1982.

\bibitem{BF77}
D.\! Buchholz and K.\! Fredenhagen.
\newblock Dilations and interaction.
\newblock {\em J. Math. Phys.}, 18(5):1107--1111, 1977.

\bibitem{BW86}
D.\! Buchholz and E.\! H.\! Wichmann.
\newblock Causal independence and the energy-level density of states in local
  quantum field theory.
\newblock {\em Comm. Math. Phys.}, 106(2):321--344, 1986.

\bibitem{CC01}
S.\! Carpi and R.\! Conti.
\newblock Classification of subsystems for local nets with trivial
  superselection structure.
\newblock {\em Comm. Math. Phys.}, 217(1):89--106, 2001.

\bibitem{Driessler79}
W.\! Driessler.
\newblock Duality and absence of locally generated superselection sectors for
  {CCR}-type algebras.
\newblock {\em Comm. Math. Phys.}, 70(3):213--220, 1979.

\bibitem{DKSS13}
A.\! Dymarsky, Z.\! Komargodski, A.\! Schwimmer and S.\! Theisen.
\newblock On scale and conformal invariance in four dimensions.
\newblock 2013.
\newblock arXiv:1309.2921.

\bibitem{GK96}
L.\! Ge and R.\! Kadison.
\newblock On tensor products for von {N}eumann algebras.
\newblock {\em Invent. Math.}, 123(3):453--466, 1996.

\bibitem{Haag96}
Rudolf Haag.
\newblock {\em Local quantum physics}.
\newblock Texts and Monographs in Physics. Springer-Verlag, Berlin, second
  edition, 1996.

\bibitem{Hislop88}
P.\! D.\! Hislop.
\newblock Conformal covariance, modular structure, and duality for local
algebras in free massless quantum field theories.
\newblock {\em Ann. Phys.}, 185(2):193--230, 1988.

\bibitem{HL82}
P.\! D.\! Hislop and R.\! Longo.
\newblock Modular structure of the local algebras associated with the free
  massless scalar field theory.
\newblock {\em Comm. Math. Phys.}, 84(1):71--85, 1982.

\bibitem{KLM01}
Y.\! Kawahigashi, R.\! Longo and M.\! M{\"u}ger.
\newblock Multi-interval subfactors and modularity of representations in
  conformal field theory.
\newblock {\em Comm. Math. Phys.}, 219(3):631--669, 2001.

\bibitem{Landau70}
L.\! J.\! Landau.
\newblock Asymptotic locality and the structure of local internal symmetries.
\newblock {\em Comm. Math. Phys.}, 17:156--176, 1970.

\bibitem{Lechner08}
G.\! Lechner.
\newblock Construction of quantum field theories with factorizing
  {$S$}-matrices.
\newblock {\em Comm. Math. Phys.}, 277(3):821--860, 2008.

\bibitem{Longo08}
R.\! Longo.
\newblock Real {H}ilbert subspaces, modular theory, {${\rm SL}(2,{\bf R})$} and
  {CFT}.
\newblock In {\em Von {N}eumann algebras in {S}ibiu: {C}onference
  {P}roceedings}, pages 33--91. Theta, Bucharest, 2008.

\bibitem{Mack77}
G.\! Mack.
\newblock All unitary ray representations of the conformal group {${\rm SU}(2,
  2)$} with positive energy.
\newblock {\em Comm. Math. Phys.}, 55(1):1--28, 1977.

\bibitem{Mandelstam83}
S.\! Mandelstam.
\newblock Light-cone superspace and the ultraviolet finiteness of the n=4
  model.
\newblock {\em Nuclear Physics B}, 213(1):149 -- 168, 1983.

\bibitem{Mund12}
J.\! Mund.
\newblock An algebraic {J}ost-{S}chroer theorem for massive theories.
\newblock {\em Comm. Math. Phys.}, 315(2):445--464, 2012.

\bibitem{Nakayama13}
Y.\! Nakayama.
\newblock A lecture note on scale invariance vs conformal invariance.
\newblock 2013.
\newblock arXiv:1302.0884.

\bibitem{NT04}
N.\! M.\! Nikolov and I.\! T.\! Todorov.
\newblock Conformal invariance and rationality in an even dimensional quantum
  field theory.
\newblock {\em Internat. J. Modern Phys. A}, 19(22):3605--3636, 2004.

\bibitem{NT01}
N.\! M.\! Nikolov and I.\! T.\! Todorov.
\newblock Rationality of conformally invariant local correlation functions on
  compactified {M}inkowski space.
\newblock {\em Comm. Math. Phys.}, 218(2):417--436, 2001.

\bibitem{RSII}
M.\! Reed and B.\! Simon.
\newblock {\em Methods of modern mathematical physics. {II}. {F}ourier
  analysis, self-adjointness}.
\newblock Academic Press, New York,
  1975.

\bibitem{RSI}
M.\! Reed and B.\! Simon.
\newblock {\em Methods of modern mathematical physics. {I}. {F}unctional analysis}.
\newblock Academic Press, New York,
  second edition, 1980.

\bibitem{Stanev13}
Y.\! Stanev.
\newblock Constraining conformal field theory with higher spin symmetry in four
dimensions.
\newblock {\em Nucl. Phys. B}, 876:651--666, 2013.

\bibitem{TakesakiII}
M.\! Takesaki.
\newblock {\em Theory of operator algebras. {II}}, volume 125 of {\em
  Encyclopaedia of Mathematical Sciences}.
\newblock Springer-Verlag, Berlin, 2003.
\newblock Operator Algebras and Non-commutative Geometry, 6.

\bibitem{Tanimoto12-2}
Y.\! Tanimoto.
\newblock Construction of {W}edge-{L}ocal {N}ets of {O}bservables {T}hrough
  {L}ongo-{W}itten {E}ndomorphisms.
\newblock {\em Comm. Math. Phys.}, 314(2):443--469, 2012.

\bibitem{Tanimoto12-1}
Y.\! Tanimoto.
\newblock Noninteraction of {W}aves in {T}wo-dimensional {C}onformal {F}ield
  {T}heory.
\newblock {\em Comm. Math. Phys.}, 314(2):419--441, 2012.

\bibitem{Tanimoto13-1}
Y.\! Tanimoto.
\newblock Construction of two-dimensional quantum field models through
  {L}ongo-{W}itten endomorphisms.
\newblock {\em Forum of Mathematics, Sigma}, 2, e7. doi:10.1017/fms.2014.3

\bibitem{Todorov06}
I.\! Todorov.
\newblock Vertex algebras and conformal field theory models in four dimensions.
\newblock {\em Fortschr. Phys.}, 54(5-6):496--504, 2006.

\bibitem{Weinberg12}
S.\! Weinberg.
\newblock Minimal fields of canonical dimensionality are free.
\newblock {\em Phys. Rev. D}, 86:105015, Nov 2012.

\bibitem{Weiner11}
Mih{\'a}ly Weiner.
\newblock An algebraic version of {H}aag's theorem.
\newblock {\em Comm. Math. Phys.}, 305(2):469--485, 2011.

\end{thebibliography}
\end{document}